\documentclass{article}

\usepackage{microtype}
\usepackage{graphicx}
\usepackage{subfigure}
\usepackage{subcaption}
\usepackage{booktabs} 
\usepackage{acronym}
\usepackage{dsfont}
\usepackage{changes}

\usepackage{hyperref}

\usepackage[preprint]{icml2026}

\usepackage{amsmath}
\usepackage{amssymb}
\usepackage{mathtools}
\usepackage{amsthm}

  \newcommand{\ie}{\textit{i.e.}}

  \newcommand{\dist}[1]{\mathcal{D}(#1)}

\newcommand{\reals}{\mathbf{R}}

  \acrodef{mdp}[MDP]{Markov decision process}

\newcommand{\Expect}{\mathds{E}}

\newcommand{\calA}{\mathcal{A}}

 \acrodef{ltl}[LTL]{Linear Temporal Logic}
  \acrodef{ltlf}[LTLf]{Linear Temporal Logic over Finite Traces}
 \acrodef{dfa}[DFA]{Deterministic Finite Automaton}
 \acrodef{pdfa}[PDFA]{Preference DFA}

\newcommand{\matO}{{\mathbf{O}}}
\newcommand{\matT}{{\mathbf{T}}}
\newcommand{\matA}{{\mathbf{A}}}

\acrodef{mdp}[MDP]{Markov decision process}

\acrodef{asw}[ASW]{Almost-Sure Winning}

\newcommand{\calS}{\mathcal{S}}

\newcommand{\calO}{\mathcal{O}}

\usepackage{mathtools}
\usepackage{mathrsfs}

\usepackage[capitalize,noabbrev]{cleveref}

\theoremstyle{plain}
\newtheorem{theorem}{Theorem}[section]

\newtheorem{lemma}[theorem]{Lemma}

\theoremstyle{definition}
\newtheorem{definition}[theorem]{Definition}
\newtheorem{assumption}[theorem]{Assumption}
\theoremstyle{remark}
\newtheorem{remark}[theorem]{Remark}

\definecolor{UForange}{RGB}{250,70,22} 

\icmltitlerunning{Submission and Formatting Instructions for ICML 2026}

\begin{document}

\twocolumn[
  \icmltitle{C-IDS: Solving Contextual POMDP via Information-Directed Objective}



  \icmlsetsymbol{equal}{*}

  \begin{icmlauthorlist}
    \icmlauthor{Chongyang Shi}{yyy}
    \icmlauthor{Michael Dorothy}{comp}
    \icmlauthor{Jie Fu}{yyy}
  \end{icmlauthorlist}

  \icmlaffiliation{yyy}{Department of Electrical and Computer Engineering, University of Florida, Gainesville, Florida 32608, USA}
  \icmlaffiliation{comp}{Aerospace Engineering, DEVCOM Army Research Laboratory, Adelphi, MD 21001, USA}

  \icmlcorrespondingauthor{Chongyang Shi}{c.shi@ufl.edu}

  \icmlkeywords{Contextual POMDP, Regret Analysis, Information Theory}

  \vskip 0.3in
]



\printAffiliationsAndNotice{}  

\begin{abstract}
We study the policy synthesis problem in contextual partially observable Markov decision processes (CPOMDPs), where the environment is governed by an unknown latent context that induces distinct POMDP dynamics. Our goal is to design a policy that simultaneously maximizes cumulative return and actively reduces uncertainty about the underlying context. 
We introduce an information-directed objective that augments reward maximization with mutual information between the latent context and the agent’s observations. We develop the C-IDS algorithm to synthesize policies that maximize the information-directed objective. We show that the objective can be interpreted as a Lagrangian relaxation of the linear information ratio and prove that the temperature parameter is an upper bound on the information ratio. Based on this characterization, we establish a sublinear Bayesian regret bound $\mathcal{O}(T R_{\max}\sqrt{K \log |\mathcal{C}|/\eta})$ over $K$ episodes.
We evaluate our approach on a continuous Light--Dark environment and show that it consistently outperforms standard POMDP solvers that treat the unknown context as a latent state variable, achieving faster context identification and higher returns.
\end{abstract}

\section{Introduction}
Partially Observable Markov Decision Processes (POMDPs) provide a principled framework for sequential decision making under uncertainty when the agent cannot directly observe the true system state. Such settings are common in autonomous driving~\cite{Brechtel2014probabilistic, Bai2015intention}, surveillance and rescue~\cite{Egorov2016Target, Wandzel2019Multi, zhang2024shrinkingpomcpframeworkrealtime}, and human-robot interaction~\cite{Chen2020Trust, Burks2023harps}, where the agent must act based on noisy observations. While POMDPs provide a   framework for reasoning under state uncertainty, they typically assume a single, fixed model of the environment. In contrast, many practical systems, such as robotic navigation under changing sensing conditions, adaptive assistance, and exploratory planning, are better described by a family of POMDPs, each corresponding to a distinct but unobserved context that governs the underlying transition, observation, and reward structures. 

To capture this setting, we study Contextual Partially Observable Markov Decision Processes (CPOMDPs), which generalize both contextual MDPs and classical POMDPs. In a CPOMDP, the environment is governed by an unobserved context variable $c$ that remains fixed within an episode and induces a distinct POMDP model. The agent does not observe the context directly and must infer it through interaction using noisy and partial observations of the system state. As a result, optimal decision making requires jointly optimizing task performance and information acquisition to identify the underlying context.

\begin{figure}[t]
  \vskip 0.2in
  \begin{center}
    \centerline{\includegraphics[width=\columnwidth]{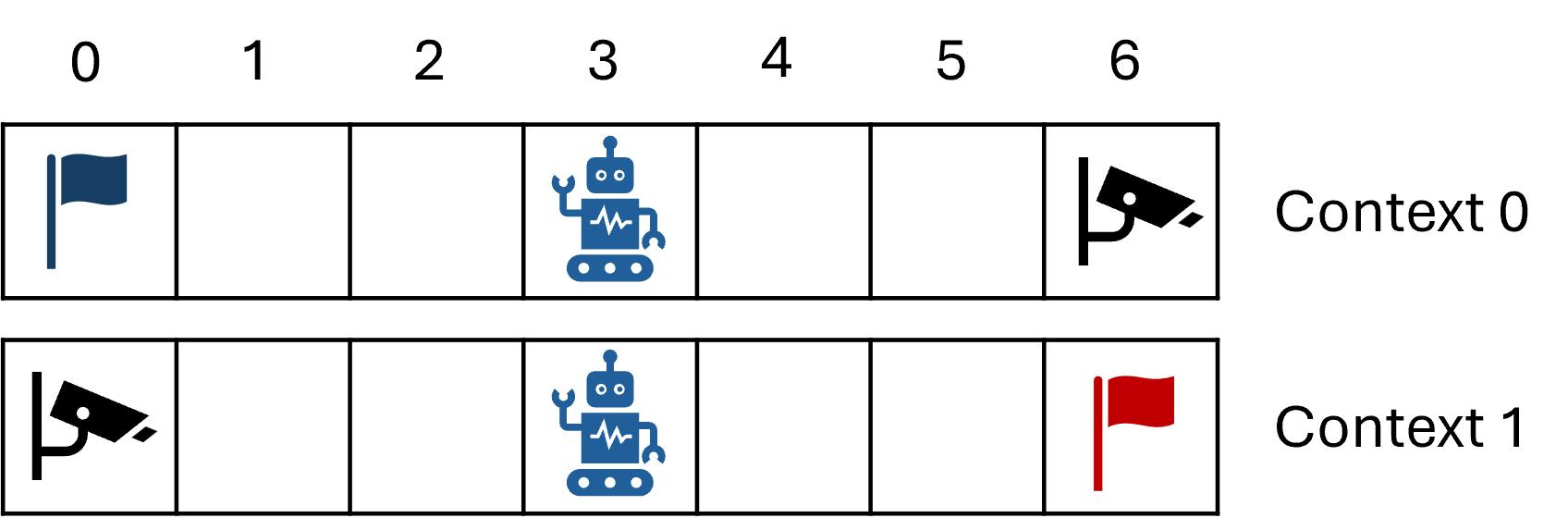}}
    \caption{ A robot moves in a line grid, with seven cells. At flag cells, the robot can receive some reward. At detector cell, the robot will be detected and receive penalty. Context 0: Cell 0 has high-value target with reward, and cell $6$ is equipped with a detector. Cell 0 has a low-value target with reward 10. Context 1: Cell $6$ has high-value target with reward, and cell $1$ is equipped with a detector. The robot can choose to move to one of the adjacent cell or the robot uses ``sense'' action to detect the presence of detector to its left or right.}
    \label{fig:line_grid}
  \end{center}
\end{figure}

Let us consider a simple line-grid environment (Figure~\ref{fig:line_grid}). An optimal policy for the robot first selects the \emph{sense} action to infer the location of the detector. Once the detector’s position is identified, the robot moves in the opposite direction to safely reach the high-value target. Specifically, if the detector is detected on the right, the robot moves left, and vice versa. In this example, successful behavior requires the robot to actively infer the context through sensing actions, as naively committing to a single direction without information gathering may result in detection penalties.

This dual objective distinguishes CPOMDPs from standard POMDPs and contextual MDPs \cite{Hallak2015CMDP}. Standard POMDP treats context as an augmented latent state and applying generic POMDP solvers often leads to undesirable policies, while contextual MDP modeling assumes full state observability, which is not satisfied for many applications. CPOMDPs therefore pose a fundamental challenge in balancing reward maximization with context inference. 

\paragraph{Contribution.}
To address this challenge, we propose an information-directed objective that augments reward maximization with a mutual information regularizer between the latent context and the agent’s observations. Building on this formulation, we develop a variational policy gradient algorithm 
that optimizes an entropy-regularized objective by matching the policy-induced trajectory distribution to a Gibbs target distribution, inspired by \cite{Levine2018ReinforcementLA}. To evaluate how this information-directed objective balances the tradeoff between exploration and exploitation in CPOMDPs, we introduce  \emph{linear information ratio}, defined by the regret per unit information gain. This metric is a variant of information ratio, traditionally defined as the squared regret per unit information gain, studied in multi-armed bandit problems \cite{Russo2014learning}. We show that this objective admits a principled interpretation through information-directed sampling: maximizing the regularized value is equivalent to minimizing a Lagrangian relaxation of the linear information ratio. Under this characterization, the regularization weight serves as an upper bound on the linear information ratio. Leveraging this connection, we establish a sublinear Bayesian regret bound for the proposed CPOMDP solver, providing formal guarantees on both exploration efficiency and task performance. Finally, we validate the effectiveness of C-IDS in a continuous Light–Dark environment~\cite{Fischer2020pf} with uncertain contexts, demonstrating faster context identification, improved returns, and significant performance gains over baseline POMDP solvers that do not explicitly reason about context uncertainty.


\paragraph{Related Work.}
Contextual decision-making has been studied extensively in reinforcement learning through the framework of \emph{Contextual Markov Decision Processes} (CMDPs). \citet{Hallak2015CMDP} introduce CMDPs to model families of MDPs indexed by a latent context that remains fixed within each episode, and propose algorithms with theoretical guarantees for learning policies that generalize across contexts. Subsequent work extends this framework by considering richer assumptions on how context influences the environment. In particular, Modi et al.~\cite{Modi2017SideInformation} study MDPs with continuous side information, where the context is observed and parameterizes the transition and reward functions, enabling efficient exploration and regret bounds under smoothness assumptions. \citet{Modi2018GLMCMDP} further develops CMDPs with structured context dependence, including generalized linear models for transition dynamics, and focuses on regret-optimal online learning.

More recently, contextual MDPs have been embedded into hierarchical and bilevel optimization frameworks. \citet{Hu2024BilevelCMDP} formulate a stochastic bilevel optimization problem in which a leader selects contextual parameters that induce a lower-level CMDP solved by a follower, and develop gradient-based methods for optimizing the upper-level objective. These bilevel CMDP formulations are motivated by mechanism design, where the context serves as a controllable   parameter affecting the lower-level MDP.

While these works extend MDPs to account for contextual variability, they differ from the CPOMDP studied in this paper. Existing CMDP formulations   assume full observability of the system state and focus on learning or optimizing policies across context-indexed MDPs. In contrast, CPOMDPs combine latent context uncertainty with partial observability of states, requiring the agent to actively infer the underlying context through noisy observations while simultaneously optimizing context-dependent task performance. 

We introduce information-directed objective to balance the trade-off between inference and reward optimization. Our work is closely related to the information-directed sampling. \citet{Russo2016TS} analyze Thompson sampling through information theory and introduce regret bounds that scale with the entropy of the optimal action, formalizing how information gain controls learning efficiency. Building on this viewpoint, information-directed sampling (IDS) selects actions by explicitly trading off expected regret against information gain, via minimizing an information ratio that compares (squared) instantaneous regret to mutual information about the optimal decision \cite{Russo2014learning}. More recently, \citet{neuOptimisticInformationDirectedSampling} propose \emph{optimistic} variants of IDS that achieve similar regret guarantees without relying on Bayesian priors, bridging Bayesian IDS with worst-case online learning analyses. Extending IDS beyond bandits, Hao et al.\ \cite{Hao2022IDRL} develop information-directed exploration principles for reinforcement learning and establish regret bounds for IDS-style algorithms in episodic MDP settings, including tractable regularized surrogates. 

While these works share the principle of balancing performance and information acquisition, they differ substantially from our setting and technical contributions. Existing analyses on IDS and information-ratio  primarily focus on bandits or fully observable MDPs, where the learning target is typically the optimal arm/action or unknown model parameters, and observations provide relatively direct feedback about that target \cite{Russo2014learning,Russo2016TS,Hao2022IDRL}. In contrast, we study CPOMDP, where the agent must estimate the latent states and infer the context through noisy observation trajectories. In CPOMDPs, minimizing the standard squared information ratio is computationally challenging due to the high-dimensional belief and observation spaces. To address this issue, we introduce a linear information ratio, which admits a more tractable optimization algorithm. By leveraging fractional programming, we establish a direct connection between minimizing the linear information ratio and maximizing the proposed information-directed objective. Importantly, despite replacing the squared information ratio with its linear counterpart, the resulting CPOMDP solver still achieves a sublinear Bayesian regret bound, demonstrating that the improved computational efficiency does not come at the expense of theoretical performance guarantees.

\section{Contextual Partially Observable Markov Decision Process}

A Contextual  Partially Observable Markov Decision Process
(CPOMDP) is a tuple  
\[
M = \langle  \mathcal{C}, \calS,  \calA, \{ M_i, i \in \mathcal{C}\}\rangle 
\]
where 
\begin{itemize}
\item $\mathcal{C} =\{1,\ldots N\}$ is a finite set of context space.
\item $\calS$ is a set of states.
\item $\calA$ is a set of   actions.
\item $M_i $ is a POMDP with state, action spaces $\calS$ and $\calA$, and  \begin{itemize}
\item $\{P_i\}$ is a finite set of transition functions. For type $i$, the interaction between the agent and its environment is defined by $P_i: \calS\times \calA \rightarrow \dist{\calS}$. 
\item $\{E_i\}$ is a finite set of emission functions of the opponent. For type $i$, the opponent's observation function is $E_i:S\rightarrow \dist{\calO}$.
\item $\{r_i\}$ is a finite set of reward functions.
\item $\mu_i$ is an initial distribution over the   state.
\end{itemize} 
\end{itemize}

A CPOMDP is simply a set of models that shareare the
same state and action space.

For any context \( c \) and policy \( \pi \), define the episodic value
\[
V_c(\pi) := \Expect_{\pi} \left[\sum_{t=0}^{T-1} r_c(S_t, A_t)\right],
\]
where $T$ is the length of horizon. 

Given two contexts $c_i, c_j \in C$, we say that they are \emph{$\epsilon$-policy invariant} if
\[
\max_{\pi} V_{c_i}(\pi) - V_{c_i}(\pi_j^\ast) \le \epsilon,
\]
and
\[
\max_{\pi} V_{c_j}(\pi) - V_{c_j}(\pi_i^\ast) \le \epsilon,
\]
where $\pi_i^\ast \in \arg\max_{\pi} V_{c_i}(\pi)$ and
$\pi_j^\ast \in \arg\max_{\pi} V_{c_j}(\pi)$
denote optimal policies under contexts $c_i$ and $c_j$, respectively.
We write $c_i \sim_\epsilon c_j$ when the two contexts are $\epsilon$-policy invariant.

\begin{assumption}
\label{assume:not-policy-invariant}
Given a constant $\epsilon > 0$, the context set $C$ is \emph{pairwise not $\epsilon$-policy invariant}; that is, for any $c_i, c_j \in C$ with $i \neq j$,
\[
c_i \not\sim_\epsilon c_j .
\]
\end{assumption}

If there is a pair of contexts in $C$ that are   $\epsilon$-policy invariant, we can always remove one of the contexts to satisfy this assumption. In the following, we consider CPOMDPs satisfying Assumption~\ref{assume:not-policy-invariant}.


\subsection{Planning to Maximize Information-Directed Objective}
 

We consider an agent that interacts repeatedly with an unknown environment over a sequence of finite-horizon episodes. In each episode, an unobserved context is sampled at the beginning and remains fixed throughout the episode, while potentially changing across episodes. Our goal is to compute a policy that maximizes the finite-horizon entropy-regularized total reward while actively reducing uncertainty about the latent context. This episodic context variability, combined with partial observability, makes learning an optimal policy particularly challenging for standard reinforcement learning algorithms that assume a fixed environment model.

At the $k$-th episode, the agent executes a policy $\pi$ for a horizon $T$ and observes a trajectory $Y^{k} := (O^{k}_{0:T}, A^{k}_{0:T-1})$,
where $O^{k}_{t}$ and $A^{k}_{t}$ denote the observation and action at time step $t$ of episode $k$, respectively. and let $\mathcal H_{k} := (Y^0,\dots,Y^{k})$ be  a history of observations.
Let $C$ be a context random variable referring to the belief about the context. The uncertainty in the context given some state-action observation $y_{0:t} = (o_{0:t}, a_{0:t-1})$ can be measured using the conditional entropy $H(C|y_{0:t})$. Note that this conditional entropy is independent from a control policy because the action sequence is included in $y_{0:t}$.

Let $P (C| \mathcal{H}_{k-1})  $ be the posterior estimate of the context given a history of observations $\mathcal H_{k-1}$. Define the posterior-expected value
\begin{equation}
\label{eq:posterior-expected}
\bar V_k(\pi)
:= \Expect_{c\sim P(C| \mathcal{H}_{k-1})} \!\left[V_c(\pi) \right],
 \end{equation}

  \paragraph{$C$-IDS: Maximizing a weighted combination of reward and information gain}
At episode $k$, the contextual IDS ($C$-IDS) algorithm selects a policy by solving
\begin{equation}
\label{eq:objective_function}
\max_{\pi}
\;\Expect_\pi\!\left[
\bar V_k(\pi)
+
\tau\, I\!\left(C; Y^k \mid \mathcal{H}_{k-1}\right)
\right],
\end{equation}
where $\tau > 0$ is a temperature parameter that balances reward maximization and context uncertainty reduction.
The expectation is taken with respect to the randomness induced by the policy $\pi$ and the environment dynamics.

Using the identity $I(X;Y \mid Z) = H(X \mid Z) - H(X \mid Y, Z)$ and $H(C|\mathcal{H}_{k-1})$ is independent from the policy $\pi$ used at the $k$-th episode, the objective in~\eqref{eq:objective_function} can be equivalently written as
\begin{equation}
\label{eq:optimize-c-ids}
\max_{\pi}
\;\Expect_\pi\!\left[
\bar V_k(\pi)
-
\tau\, H\!\left(C \mid \mathcal{H}_{k}\right)
\right].
\end{equation}





\subsection{Variational Policy Gradient for Optimizing Information-Directed Objective}

We consider a variational approach to solve  the optimal policy in \eqref{eq:optimize-c-ids}. Suppose we can directly select the joint distribution over context and observation, $Q(C, Y_{0:T})$, the following result indicates how to make this choice.

\begin{lemma}
\label{lem:Gibbs}
The entropy-regulated optimal  distribution $Q$   that maximizes  \[
Q^\ast = \arg\max_{Q} \left( \Expect_{(y,c)\sim Q}\left[ R_c(y) + \tau \log P(c|y) \right] + \tau H(Q) \right),
\]
where $ H(Q)$ is the Shannon entropy of $Q$, 
is given by:
\[
Q^\ast(y,c) \propto P(c|y) \exp\left(\frac{R_c(y)}{\tau}\right)
\]
where $ R_c(y) = \Expect_i \left[\sum_{t=0}^T R_c(S_t, A_t)|y\right]$.
\end{lemma}
\begin{proof}
See Appendix~\ref{app:proof}.
\end{proof}

Note that the normalization constant
\[
Z := \sum_{c} \sum_{y} P(C = c \mid y)\exp\!\left( \frac{R_c(y)}{\tau} \right)
\]
is independent of the policy $\pi$, since the action sequence contained in $y$ is fixed.

\begin{equation}
\label{eq:target_distribution}
Q^\ast(c, y)
=
\frac{
P(C = c \mid y)\exp\!\left( \frac{R_c(y)}{\tau} \right)
}{
Z
}.
\end{equation}

Taking $Q^\ast(Y)$ as the target distribution, the variational approach seeks a policy $\pi$ that minimizes the divergence between $P_\pi(Y)$ and $Q^\ast(Y)$, that is,
\begin{equation}
\label{eq:vpg_min}
\pi \in \arg\min_{\pi } = D_{KL}(P_\pi \| Q)
\end{equation}
where
\[
P_\pi(Y = y)
=
\sum_{c} P(Y = y \mid C = c, \pi)\, P(C = c).
\]

\begin{lemma}
\label{lem:vpg}
Let $\theta$ denote the parameters of the policy $\pi_\theta$ (or any model inducing $P_\theta$).
The KL divergence is 
\begin{equation}
\mathcal{L}(\theta) = D_{KL}(P_\theta \| Q) \\
= \Expect_{(y, c) \sim P_\theta} \!\left[\log \frac{P_\theta(y,c)}{Q(y,c)}\right]
\end{equation}
The gradient is 
\begin{equation}
\label{eq:variational_gradient}
\nabla_\theta \mathcal{L}(\theta) = \Expect_{(y, c) \sim P_\theta} \!\left[\left(\log P_\theta(y) - \frac{R_c(y)}{\tau} \right) \nabla_\theta \log P_\theta(y|c) \right]
\end{equation}
\end{lemma}
\begin{proof}
See Appendix~\ref{app:proof}.
\end{proof}




\section{Regret Analysis of $C$-IDS Algorithm}


 We consider an episodic setting in which an agent interacts with an unknown environment modeled as a CPOMDP with an unknown context over multiple episodes. The latent context is assumed to remain fixed across all episodes. Using the $C$-IDS algorithm (Algorithm~\ref{alg:regret_algorithm}), the agent chooses a sequence of policies that optimize the information-directed objective \eqref{eq:optimize-c-ids} given the current posterior estimate in each episode. 
Our goal  is to show that, by doing so,  the agent can  achieve vanishing regret over episodes.

Let \( c^\star \) denote the true (but unknown) context and
\[
\pi^\star \in \arg\max_\pi V_{c^\star}(\pi)
\]
denote the optimal policy if the context were known.

Let $\pi_1,\pi_2,\ldots, \pi_K$ be a sequence of policies selected in $K$ episodes. The Bayesian regret  after \( K \) episodes is defined as
\[
\mathrm{BR}(K,T)
=
\sum_{k=1}^K
\big(
V_{c^\star}(\pi^\star)
-
V_{c^\star}(\pi_k)
\big).
\]

At episode $k$,  he posterior-expected value is  $\bar V_k 
$  (see \eqref{eq:posterior-expected})
  and let the maximal posterior-expected value be
\[
\bar V_k^\star := \max_\pi \bar V_k(\pi).
\]

Next, we perform the regret analysis for $C$-IDS  algorithm.  

Define the expected episodic regret
\[
\Delta_k(\pi) := \bar V_k^\star - \bar V_k(\pi),
\]
and the information gain at episode $k$
\[
I_k(\pi) := I(C;Y^k \mid \mathcal H_{k-1},\pi).
\]



Recall that $C$-IDS  algorithm yields a sequence of policies $\pi_1,\pi_2,\ldots, \pi_k$ where at episode $k$, the policy $\pi_k$         maximizes the information-directed objective (see also Remark~\ref{rmk})
 \[
  \pi_k = \arg\max_{\pi \in \Pi} \bigl[ \bar V_k(\pi) + \tau I_k(\pi) \bigr].
  \]
For each episode \( k \), decompose the regret as
\begin{equation}
\label{eq:decompostion}
\begin{aligned}
V_{c^\star}(\pi^\star) - V_{c^\star}(\pi_k)
&= \underbrace{
\bar V_k^\star - \bar V_k(\pi_k)
}_{\Delta_k} + \underbrace{
V_{c^\star}(\pi^\star) - \bar V_k^\star
}_{I_k^{(1)}}.\\
& +\underbrace{
\bar V_k(\pi_k) - V_{c^\star}(\pi_k)
}_{I_k^{(2)}} 
\end{aligned}
\end{equation}

\begin{algorithm}[t]
\caption{C-IDS Algorithm}
\label{alg:regret_algorithm}
\begin{algorithmic}[1]
\REQUIRE Context set $\mathcal C$, policy class $\Pi_\theta$, horizon $T$, temperature $\tau$, number of episodes $K$

\STATE Sample a true context $c^{\star} \sim P(c)$

\WHILE{$\Delta_k(\pi_k) > 0$}
    \STATE Sample context $c^{k} \sim P(C|\mathcal{H}_{k-1})$
    \STATE Execute policy $\pi_{k-1}$ in context $c^{k}$ to obtain observations
    \[
        y^{k} = (o_0^{k}, a_0^{k}, \dots, o_T^{k})
    \]
    \STATE Update posterior belief $P(C \mid \mathcal{H}_{k-1})$
    \STATE Compute optimal policy
    \[
    \pi_k \in \arg\max_\pi
    \Big(
    \bar V_k(\pi)
    +
    \tau I(C;Y^k \mid \mathcal H_{k-1},\pi)
    \Big).
    \]
\ENDWHILE
\STATE \textbf{return} $\pi^\star \coloneqq \pi_{k}$
\end{algorithmic}
\end{algorithm}

In the following subsections, we will provide a bound for each term in the decomposed regret. 
The following assumptions are made: 
\begin{assumption}[Bounded Rewards]
There exists a constant \( R_{\max} > 0 \) such that for all \( c \in \mathcal C \), \( s \in \mathcal S \), and \( a \in \mathcal A \),
\[
|r_c(s,a)| \le R_{\max}.
\]
\end{assumption}

This implies that for all \( c \) and \( \pi \),
\[
|V_c(\pi)| \le T R_{\max}.
\]

\begin{assumption}[Optimal Policy]
\label{assume:optimal}
For each episode $k$, there exists a Bayesian optimal policy $\pi_k^\star \in \Pi$ that incurs zero instantaneous regret, i.e.,
\[
\Delta_k(\pi_k^\star) = 0,
\]
where $\Pi$ denotes a given policy class.
\end{assumption}
\begin{remark}
\label{rmk}
In practice, the policy $\pi_k$ is obtained by  the variational policy gradient algorithm, interacting with a simulator that samples the context from the posterior distribution $P(C|\mathcal{H}_{k-1})$ in each iteration. When the policy converges, the converged policy is then applied to the true unknown environment to collect a new observation $Y^{k}$. 
As a result,  the following regret  is evaluated over  episodes interacting with the true CPOMDP rather than with the simulator.
\end{remark}

\subsection{Bound for Expected Episodic Regret $\Delta_k$}

First, we will prove that maximizing the objective function~\eqref{eq:objective_function} is equivalent to minimizing the information ratio. The proof was inspired by a standard solution for the fractional programming called Dinkelbach's method~\cite{Dinkelbach1967}. 

The \textit{linear information ratio} is defined as
\[
\Psi_k(\pi) := \frac{\Delta_k(\pi)}{I_k(\pi)},
\qquad I_k(\pi)>0.
\]

Directly minimizing $\Psi_k(\pi)$ is a fractional program.
Note that C-IDS optimizes the entropy-regularized objective,
\begin{equation*}
\begin{aligned}
J_k(\pi;\tau)
& := \bar V_k(\pi) + \tau I_k(\pi) \\
&= \bar V_k^\star + \bar V_k(\pi) - \bar V_k^\star  + \tau I_k(\pi) \\
&=  \bar V_k^\star + \tau I_k(\pi) - \Delta_k(\pi)
\end{aligned}
\end{equation*}
for $\tau>0$. 
Since the maximal posterior-expected value $\bar V_k^\star$ is independent of $\pi$, maximizing $J_k$ is equivalent to minimizing $\Delta_k(\pi) - \tau I_k(\pi)$. In other words, the policy $\pi_k$ selected at the $k$-th episode satisfies
\begin{equation*}
\label{eq:minus_form}
\pi_k \in \arg\min_\pi \big(\Delta_k(\pi) - \tau I_k(\pi)\big).
\end{equation*}

 \begin{lemma}
 \label{lem:upper-bound}
 For any $
\pi_k \in \arg\min_\pi \big(\Delta_k(\pi) - \tau I_k(\pi)\big)$, 
\[  \inf_{\pi: I_k(\pi)>0}
\frac{\Delta_k(\pi)}{I_k(\pi)} \le \Psi_k(\pi_k)
 \le \tau 
\]
whenever $I_k(\pi_k)>0$.
 \end{lemma}
 To prove this Lemma, we first show some result from 
ratio minimization. 
\begin{lemma}[Lagrangian surrogate for ratio minimization]
\label{lem:Lagrangian}
Let
\[
\rho_k^\star := \inf_{\pi: I_k(\pi)>0}
\frac{\Delta_k(\pi)}{I_k(\pi)}.
\]
For any $\lambda > 0$, the following are equivalent:
\begin{enumerate}
  \item $\rho_k^\star \le \lambda$,
  \item $\inf_\pi \big(\Delta_k(\pi) - \lambda I_k(\pi)\big) \le 0$.
\end{enumerate}
\end{lemma}

Thus, minimizing $\Delta_k(\pi)-\tau I_k(\pi)$ is a Lagrangian relaxation of minimizing
the information ratio $\Delta_k(\pi)/I_k(\pi)$ where $\tau$ is the  Lagrangian multiplier. 

With this result, we provide the proof of Lemma~\ref{lem:upper-bound}.

\begin{proof}[Proof of Lemma~\ref{lem:upper-bound}] Let $\pi_k^\star \in \arg\max_\pi \bar V_k(\pi)$. 
By Assumption~\ref{assume:optimal}, $\pi_k$ minimizes $\Delta_k(\pi)-\tau I_k(\pi)$, i.e., 
\[
\Delta_k(\pi_k) - \tau I_k(\pi_k)
\;\le\;
\Delta_k(\pi_k^\star) - \tau I_k(\pi_k^\star),
\]
By Assumption~\ref{assume:optimal} of the expected episodic regret, $\pi_k^\star$  satisfies
$\Delta_k(\pi_k^\star)=0$. By definition $\tau > 0$ and mutual information $I_k(\pi) \ge 0$, we have
\[
\Delta_k(\pi_k) - \tau I_k(\pi_k)
\;\le\;
-\tau I_k(\pi_k^\star) \;\le\; 0.
\]
Then  
\[
\Delta_k(\pi_k) \le \tau I_k(\pi_k).
\]
Whenever $I_k(\pi_k)>0$, this implies
\[
\Psi_k(\pi_k)
=
\frac{\Delta_k(\pi_k)}{I_k(\pi_k)}
\le
\tau.
\]
and by Lemma~\ref{lem:Lagrangian}, we have $$ \rho_k^\star := \inf_{\pi: I_k(\pi)>0}
\frac{\Delta_k(\pi)}{I_k(\pi)} \le \Psi_k(\pi_k)
 \le \tau.$$
\end{proof}

  Lemma~\ref{lem:upper-bound} indicates that by minimizing the information-directed objective function, C-IDS provides a bound $\tau$ on the minimal linear information ratio $\rho^\star_k$.



By the chain rule of mutual information,
\[
\sum_{k=1}^K I(C;Y^k \mid \mathcal H_{k-1})
=
I(C;Y^k)
\le
H(C)
\le
\log|\mathcal C|.
\]

Hence,
\begin{equation}
\label{eq:instant_regret_bound}
\sum_{k=1}^K \Delta_k
\le
\tau \log|\mathcal C|.
\end{equation}

\subsection{Bounding the Term $I_k^{(1)}$}
\begin{lemma}
\label{lem:I1}
The summation of $I_k^{(1)}$ has an upper bound
\begin{equation}
\label{eq:i1_bound}
\sum_{k=1}^K I_k^{(1)} \le \xi T R_{\max} \sqrt{\frac{K \log|\mathcal C|}{\eta}}.
\end{equation}
\end{lemma}

\begin{proof}
See Appendix~\ref{app:proof}.
\end{proof}

\subsection{Bounding the Martingale Difference $I_k^{(2)}$}

\begin{lemma}
\label{lem:I2}
The expectation of summation of martingale difference term $I_k^{(2)}$ has an upper bound
\begin{equation}
\label{eq:i2_bound}
\Expect \left|\sum_{k=1}^K I_k^{(2)}\right|\le 2TR_{\max}\sqrt{2K}
\end{equation}
\end{lemma}

\begin{proof}
See Appendix~\ref{app:proof}.
\end{proof}

\subsection{Final Regret Bound}

\begin{theorem}[Bayesian Regret of Entropy-Regularized CPOMDP]
Consider a CPOMDP with finite context set $\mathcal C$, episodic horizon $T$, and true but unknown context $c^\star$.
Let the policy $\pi_k$ be chosen according to the entropy-regularized objective
\[
\pi_k \in \arg\max_\pi
\Big(
\bar V_k(\pi)
+
\tau I(C;Y^k \mid \mathcal H_{k-1},\pi)
\Big).
\]

Then the expected Bayesian regret after $K$ episodes satisfies
\begin{equation*}
\begin{aligned}
&\Expect\!\left[\mathrm{BR}(K,T)\right]
:=
\Expect\!\left[
\sum_{k=1}^K
\big(
V_{c^\star}(\pi_{c^\star}^\star)
-
V_{c^\star}(\pi_k)
\big)
\right] \\
& \;\le\;
\tau \log|\mathcal C| +
\xi T R_{\max}\sqrt{\frac{K\log|\mathcal C|}{\eta}} +
2 T R_{\max}\sqrt{KT},
\end{aligned}
\end{equation*}
for universal constants $\xi>0$.
In particular,
\[
\mathrm{BR}(K,T)
=
O\!\left(
T R_{\max}\sqrt{\frac{K\log|\mathcal C|}{\eta}}
\right).
\]
\end{theorem}

\begin{proof}
The result follows directly from the decomposition in~\eqref{eq:decompostion}. By combining the bounds in ~\eqref{eq:instant_regret_bound},~\eqref{eq:i1_bound},~\eqref{eq:i2_bound} and applying the triangle inequality, the theorem is established.
\end{proof}


\begin{figure}[t]
  \vskip 0.2in
  \begin{center}
    \centerline{\includegraphics[width=\columnwidth]{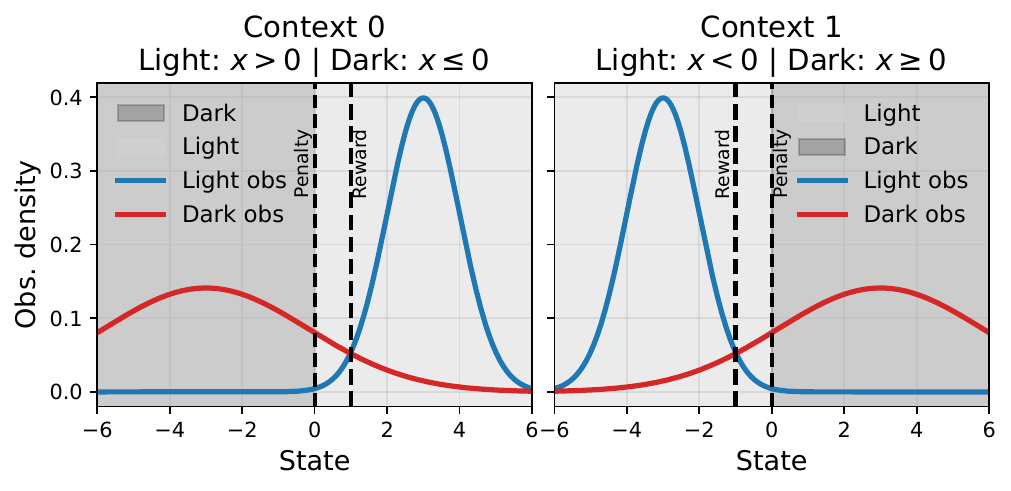}}
    \caption{The pictures show the light and dark environment. In context $0$, the light region is $x > 0$ and the dark region is $x \le 0$. In context $1$, the light region is $x < 0$ and the dark region is $x \ge 0$. In different regions, the agent has different observations noise. The observation models are Gaussian distributions shown by red and blue curves in the pictures. In context $0$, the reward region is $x > 1$ and the penalty region is $x < 1$. In context $1$, the reward region is $x < -1$ and the penalty region is $x > 0$.}
    \label{fig:env}
  \end{center}
\end{figure}

\begin{figure*}[t]
  \vskip 0.2in
  \centering

  \subfigure[Total return\label{fig:rewards}]{
    \includegraphics[width=0.31\textwidth]{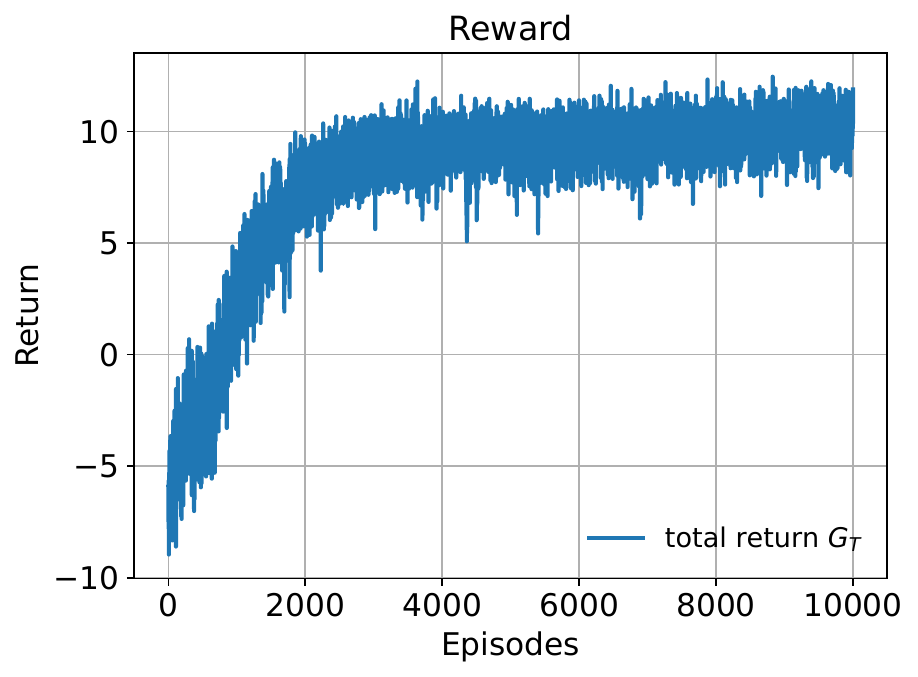}
  }
  \hfill
  \subfigure[Entropy\label{fig:entropy}]{
    \includegraphics[width=0.31\textwidth]{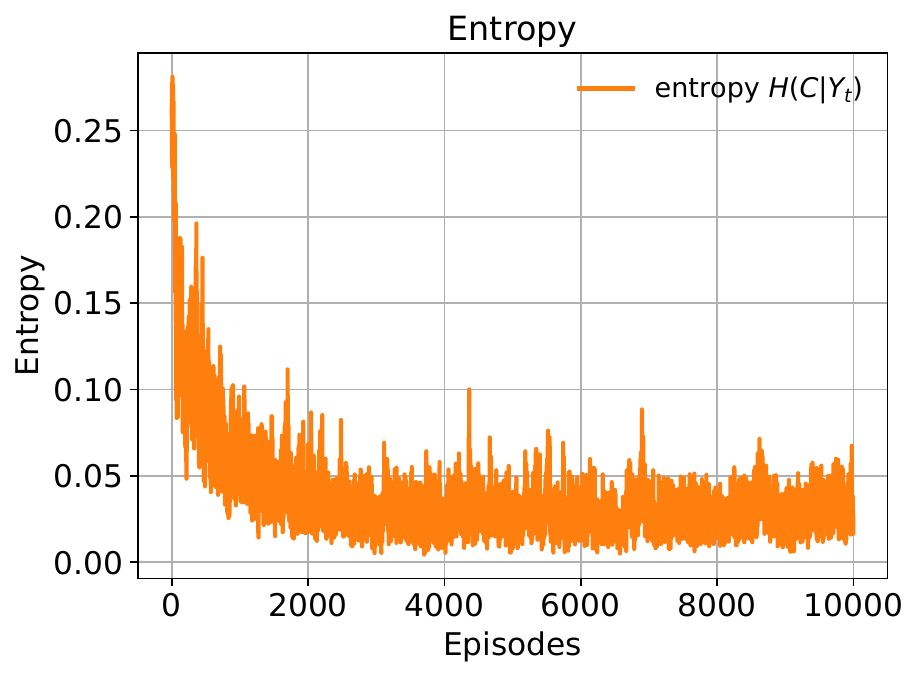}
  }
  \hfill
  \subfigure[Regret\label{fig:regret}]{
    \includegraphics[width=0.31\textwidth]{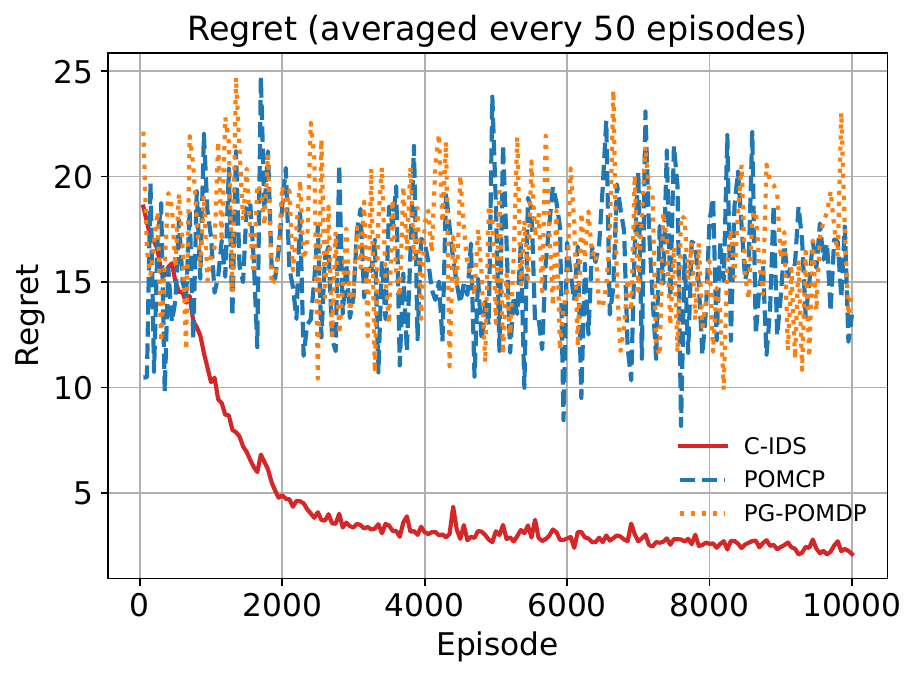}
  }

  \caption{Convergence results for total return, entropy, and regret for different algorithms.}
  \label{fig:convergence}
\end{figure*}

\section{Experiments}
\paragraph{Environment.}
To illustrate the effectiveness of the C-IDS algorithm using variational policy gradient solver, we consider a one-dimensional continuous-state Light--Dark environment~\cite{Fischer2020pf} with a latent context variable $c \in \{0,1\}$ that determines the spatial structure of informative regions. The goal of the agent is to reach the reward region shown in the Figure~\ref{fig:env} and stay there. In particular, the agent receives reward $1$ if it is in the reward region and receive the penalty $-1$ if it is in the penalty region. 
In each episode, the agent does not observe the immediate reward, but can only receive the cumulative discounted reward at the end of an episode.
This example captures the interaction between goal-reaching and information acquisition.

The initial state $s_0$ follows a Gaussian distribution $ \mathcal N(0,\sigma_u^2)$. The state is $x_t \in \mathbb{R}$ and the action space is
$\mathcal A=\{\texttt{l},\texttt{r},\texttt{o}\}$,  which correspond to moving left, moving right, and actively observing the state, respectively.
Given step size $\texttt{s}>0$ and dynamical noise variance $\sigma_p^2$, the dynamics are
\[
x_{t+1} = x_t + u(a_t) + \varepsilon_t,
\quad
\varepsilon_t \sim \mathcal N(0,\sigma_p^2),
\]
where
\[
u(a_t)=
\begin{cases}
-\texttt{s}, & a_t=\texttt{l},\\
+\texttt{s}, & a_t=\texttt{r},\\
0,  & a_t=\texttt{o}.
\end{cases}
\]
and $\varepsilon_t = 0$ if $a_t = \texttt{o}$. 

Observations are available only when $a_t=\texttt{o}$.
Let $\sigma_L^2$ and $\sigma_D^2$ denote the observation variances in the light
and dark regions, respectively. The observation variance is higher in the dark region than that in the light region, \ie, $\sigma^2_D >\sigma_L^2$.

The state-dependent observation noise variance is
\[
\sigma^2(x,c)=
\begin{cases}
\sigma_L^2, & (c=0,\ x>0) \; \lor \; (c=1,\ x<0),\\
\sigma_D^2, & (c=0,\ x\le 0) \; \lor \; (c=1,\ x\ge 0). \\
\end{cases}
\]
The observation model is
\[
z_{t+1} \mid x_{t+1}, c
\sim
\mathcal N\!\big(x_{t+1},\,\sigma^2(x_{t+1},c)\big),
\]
if $a_t = \texttt{o}$ and no observation is received when $a_t\in\{\texttt{l},\texttt{r}\}$.

Intuitively, to achieve a higher return, the agent must infer the underlying context and then decide whether to move left or right accordingly. One way to infer the context is to take the observation action $\texttt{o}$ to collect measurements and estimate the sample variance. However, such an observe-then-move strategy may be suboptimal, as time spent gathering observations delays progress toward the goal and incurs an opportunity cost. We implement the proposed solution with information-directed objectives.

\paragraph{Implementation.} We employ a long short-term memory (LSTM) neural network-based policy that directly processes observation sequences to action decisions for the agent. The Pytorch will automatically calculate the policy gradient term $\nabla_\theta \log  \pi_\theta(a_{t}| o_{0:t-1})$. To calculate the term $\log P_\theta(y)$, we use the extended Kalman
filter mechanism~\cite{gelb1974applied}. See the details in the Appendix~\ref{subapp:EKF}. 

In the experiment, the horizon of each episode is $T=20$.  We set the dynamical noise $\sigma_p^2 = 0.1$, the observation noises $\sigma_L^2 = 1$ and $\sigma_D^2 = 8$ and the initial state distribution variance $\sigma_u^2 = 0.09$. The temperature $\tau = 0.2$. During training, we set the prior distribution over contexts to be uniform, with $P(C = 0) = P(C = 1) = 0.5$. To obtain the policy in the $k$-th episode in the real environment, we employ variational policy gradient in a simulator with a fixed distribution $P_k(C|\mathcal{H}_{k-1})$ about the context. 
Figure~\ref{fig:rewards} and~\ref{fig:entropy} illustrates the convergence trend of the variational policy gradient algorithm \footnote{See the code in the supplementary material. The step size $\texttt{s} = 1.0$. We sample $M = 200$ trajectories for gradient approximation  for each iteration. The fixed learning rate of the policy gradient algorithm is set to be $0.001$. The hidden dimensions of all layers are set to be $64$. We run $N = 10000$ episodes for training on the 12th Gen Intel(R) Core(TM) i7-12700. The average time consumed for one iteration is $7.6$ seconds.}

Upon convergence, the conditional entropy $H(C | Y; \theta^\star)$ is approximately $0.03$, which is close to zero. This indicates that the observations provide significant information about the environment context. Furthermore, the total return under the optimal policy, $G_T^\star$, reaches approximately $10$, suggesting that the agent can  receive an optimal total return regardless of the environment context. The red line of Figure~\ref{fig:regret} shows that the regret decreases quickly in the first $2000$ episodes, which shows the efficiency of C-IDS algorithm. See visualization results in Appendix~\ref{app:visualization}.


\begin{figure}[t]
  \vskip 0.2in
  \begin{center}
    \centerline{\includegraphics[width=0.8\columnwidth]{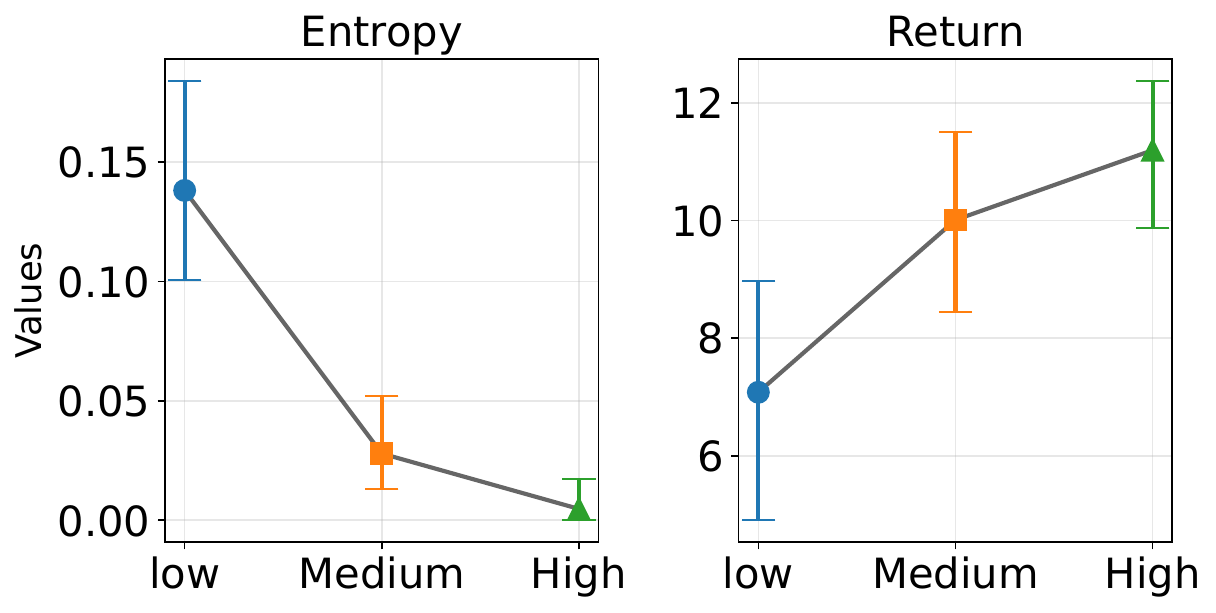}}
    \caption{The error bars for different variance ratio.}
    \label{fig:error_bars}
  \end{center}
\end{figure}

\paragraph{Result Discussion.} 
We observe that the agent identifies the underlying context by exploiting differences in observation noise. In particular, when the agent experiences larger observation noise, it infers that its current state lies in the dark region; otherwise, it is more likely to be in the light region. Figure~\ref{fig:error_bars} illustrates this behavior, where the ratio denotes the observation noise variance ratio $\sigma_D^2 : \sigma_L^2$, with values $4\!:\!1$, $8\!:\!1$, and $16\!:\!1$, respectively. As this ratio increases, the distinction between light and dark regions becomes more pronounced, making the context easier to identify due to the increased separability in observation variance. 
Moreover, Figure~\ref{fig:error_bars} shows that reducing context uncertainty directly improves performance: lower entropy corresponds to higher cumulative return, highlighting the benefit of active information acquisition for reward optimization. 


\begin{table}[t]
\centering
\begin{tabular}{@{}ccc@{}}
\toprule
 & Total Return & Entropy \\ \midrule
C-IDS & 10.572 ± 1.345 & 0.020 ± 0.015 \\
RDPG-RNN & -3.565 ± 3.045 & 0.164 ± 0.002 \\
POMCP & -6.504 ± 2.023 & 1.000 ± 0.000 \\ \bottomrule
\end{tabular}
\caption{The baseline comparison results.}
\label{tab:baseline}
\end{table}

\paragraph{Baseline Comparison.} 
We compare C-IDS against two standard POMDP solvers, POMCP~\cite{Silver2010pomcp, Fischer2020pf} and RDPG-RNN~\cite{heess2015memorybasedcontrolrecurrentneural}. To use POMDP solver, we consider the context is a latent state variable and then modeled the contextual POMDPs as an augmented-state POMDP with state $(s, c)$. Each algorithm is trained for $10{,}000$ episodes. As shown in Figure~\ref{fig:regret}, C-IDS exhibits significantly faster regret convergence than both baselines. This advantage arises because the baseline methods fail to reliably identify the latent context during training.
To further evaluate policy quality, we test the learned policies from all three algorithms on $200$ trajectories sampled from the same environment. The results are summarized in Table~\ref{tab:baseline}. C-IDS achieves the highest performance, with an average total return of $10.572$ and a final context entropy of $0.02$, indicating near-optimal context identification. In contrast, POMCP attains a total return of $-6.504$ with an entropy of $1.000$, suggesting that it fails to infer the context altogether and consequently performs poorly. RDPG-RNN yields a relatively better outcome, achieving a total return of $-3.565$ and an entropy of $0.164$, which indicates partial context identification but still falls short of C-IDS.  See visualization in Appendix~\ref{app:visualization}.

\section{Conclusion and Future Work}

In this paper, we investigated the planning problem in CPOMDP, where an agent with partial observations must simultaneously accomplish a task and infer an unknown latent context that governs the environment. We proposed contextual information direction sampling (C-IDS) algorithm, which, at each episode, selects a policy by maximizing an information-directed objective given the current posterior belief over the latent context while interacting with an unknown environment. Due to the structure of this objective, we develop a variational policy gradient method to compute the optimal policy used in each episode of the C-IDS using a simulator of the CPOMDP and the current posterior belief. 

From a theoretical standpoint, we prove that the proposed method achieves a sublinear Bayesian regret bound by explicitly connecting the C-IDS objective to a Lagrangian relaxation of the information ratio. Further, we  decompose the cumulative regret into three interpretable components, each capturing a distinct source of performance loss. We then bound these terms separately using  Dinkelbach’s method, Pinsker’s inequality with smoothing, and the Azuma–Hoeffding inequality. Together, these bounds establish the desired sublinear regret guarantee.
Empirically, experiments in a continuous Light–Dark environment demonstrate that C-IDS efficiently identifies the latent context, achieves near-optimal returns, and significantly outperforms standard POMDP solvers that do not explicitly reason about context uncertainty.

There are several promising directions for future work. 
First, while we focused on discrete action spaces and low-dimensional continuous states, scaling the method to high-dimensional continuous observations, such as images, LiDAR, or multimodal sensor data, remains an important challenge for real-world robotic systems. Secondly, extending the proposed framework to multi-agent settings, where multiple agents must jointly infer shared or private contexts. For example, through consensus-based belief sharing or decentralized information aggregation.
\section*{Impact Statement}
This paper presents work whose goal is to advance the field of machine learning. There are many potential societal consequences of our work, none of which we feel must be specifically highlighted here.

\bibliography{refs}
\bibliographystyle{icml2026}

\newpage
\appendix
\onecolumn

\section{Proofs of Lemmas}
\label{app:proof}

\paragraph{Proof of Lemma~\ref{lem:Gibbs}.}  Follow the definition of objective function, we have 
\begin{align*}
&\mathbb{E}_{Y_{0:T} \sim Q}\!\left[ r(Y_{0:T}) - \tau H(C \mid Y_{0:T}) \right] 
= \; \sum_{y} Q(y)\left[ r(y) - \tau H(C \mid y) \right] 
= \; \sum_{y} \sum_{c} Q(y) P(c \mid y)\, R_c(y) \\
&\quad + \tau \sum_{y} \sum_{c} Q(y) P(c \mid y)\, \log P(C = c \mid Y_{0:T}) 
= \; \sum_{y} \sum_{c} Q(c, y)\left[ R_c(y) + \tau \log P(C = c \mid y) \right].
\end{align*}
The entropy-regularized optimal distribution $Q^\ast(c,y)$ is therefore given by
\begin{equation*}
\begin{aligned}
Q^\ast(c, y)
 \propto \exp\!\left( \frac{R_c(y) + \tau \log P(C = c \mid y)}{\tau} \right) 
= P(C = c \mid y)\exp\!\left( \frac{R_c(y)}{\tau} \right).
\end{aligned}
\end{equation*}

\paragraph{Proof of Lemma~\ref{lem:vpg}.} By substituting the expression~\eqref{eq:target_distribution} of the target distribution into the above equation, we obtain
\begin{equation}
\mathcal{L}(\theta) 
= \Expect_{(y, c) \sim P_\theta} \!\left[\log \frac{P_\theta(y)P(c|y)}{P(c|y)\exp(R_c(y)/\tau)/Z}\right]\\ =\Expect_{(y, c) \sim P_\theta} \!\left[ \log P_\theta(y) - \frac{R_c(y)}{\tau} + \log Z \right]
\end{equation}
The gradient is
\begin{equation}
\nabla_\theta \mathcal{L}(\theta) = \Expect_{(y, c) \sim P_\theta} \!\left[ \left(\log P_\theta(y) - \frac{R_c(y)}{\tau} \right) \nabla_\theta \log P_\theta(y, c) \right]
\end{equation}
by applying likelihood ratio trick~\cite{Williams1992simple}. Also note that $\nabla_\theta \log P_\theta(Y =y, C=c) = \nabla_\theta \log [P_\theta(Y=y | C=c)P(C=c)] = \nabla_\theta \log P_\theta(Y=y | C=c)$. Therefore,
\begin{equation}
\nabla_\theta \mathcal{L}(\theta) = \Expect_{(y, c) \sim P_\theta} \!\left[\left(\log P_\theta(y) - \frac{R_c(y)}{\tau} \right) \nabla_\theta \log P_\theta(y|c) \right]
\end{equation}

According to~\cite{wei2025activeinferenceincentivedesign, shi2025integratedcontrolactiveperception}, the gradient of probabaility of observations given the context $c$ is 
\begin{equation}
\label{eq:log_observations_gradient}
\begin{aligned}
\nabla_\theta \log  P_\theta(y|c)=  \sum_{t = 0}^{T} 
\nabla_\theta \log  \pi_\theta(a_{t}| o_{0:t-1}).
\end{aligned}
\end{equation}
The probability of observations over all contexts is
\begin{equation}
\label{eq:probability_observations}
P_\theta(y) = \sum_i  P_\theta(y |c) P(c).
\end{equation}
Then we can obtain the value of the term $\log P_\theta(y)$. By equation~\eqref{eq:probability_observations}, ~\eqref{eq:log_observations_gradient}, we can calculate the gradient~\eqref{eq:variational_gradient}. 

Notably, the calculation of $P_\theta(y |c)$ requires some inference mechanism which depend on the applications. For example, in the discrete environment, we can use observable operator~\cite{spanczerObservableOperatorModels2016}. In the  continuous systems, we can use extended Kalman filter~\cite{gelb1974applied} or particle filter~\cite{doucet2001smc}. We list all used inference mechanism in the Appendix~\ref{app:inference}. Finally, Algorithm~\ref{alg:cpomdp_solver} (Appendix~\ref{app:summary}) summarizes the variational policy gradient method for solving the optimal policy in \eqref{eq:vpg_min} under a fixed prior distribution $P(C)$ of the context.

\paragraph{Proof of Lemma~\ref{lem:Lagrangian}.} 
If $\rho_k^\star \le \lambda$, then there exists $\pi$ such that
\[
\frac{\Delta_k(\pi)}{I_k(\pi)} \le \lambda,
\]
which implies $\inf_\pi \big(\Delta_k(\pi) - \lambda I_k(\pi)\big) \le 0$ because $I_k(\pi)>0$.
Conversely, if $\inf_\pi (\Delta_k(\pi)-\lambda I_k(\pi)) \le 0$, then there exists $\pi$
with $\Delta_k(\pi) \le \lambda I_k(\pi)$.
Since $I_k(\pi)>0$, this implies $\Delta_k(\pi)/I_k(\pi) \le \lambda$.

\paragraph{Proof of Proposition~\ref{lem:I1}}
Let $P_k(c)=\mathbb P(C=c\mid\mathcal H_{k - 1})$.
Since $\bar V_k^\star$ is the maximal posterior expected value,
\[
\bar V_k^\star
\ge
\sum_{c\in\mathcal C} P_k(c)\, V_c(\pi_{c^\star}^\star).
\]
Hence
\begin{align*}
I_k^{(1)} = &  V_{c^\star}(\pi_{c^\star}^\star) - \bar V_k^\star \\
= & \sum_{c\in\mathcal C} P_k(c)\, V_{c^\star}(\pi_{c^\star}^\star) - \bar V_k^\star \\
= &  \sum_{c\in\mathcal C} P_k(c)\, V_{c^\star}(\pi_{c^\star}^\star) - \max_{\pi} \sum_{c\in\mathcal C} P_k(c)V_c(\pi)\\
\le & \sum_{c\in\mathcal C} P_k(c)\, V_{c^\star}(\pi_{c^\star}^\star)  - \sum_{c\in\mathcal C} P_k(c)V_c(\pi_{c^\star}^\star) \\
= & 
\sum_{c\in\mathcal C} P_k(c)
\big(
V_{c^\star}(\pi_{c^\star}^\star)
-
V_c(\pi_{c^\star}^\star)
\big).
\end{align*} 
By bounded per-step rewards,
\[
|V_c(\pi)-V_{c'}(\pi)| \le 2T R_{\max},
\]
so
\begin{equation}
\label{eq:i1_1mp_bound}
\begin{aligned}
I_{k}^{(1)}
&\le \sum_{c\neq c^\star} P_k(c)\cdot 2T R_{\max} \\
&= 2T R_{\max} \sum_{c\neq c^\star} P_k(c) \\
&= 2T R_{\max}\big(1 - P_k(c^\star)\big).
\end{aligned}
\end{equation}

We bound the posterior error mass using Pinsker's inequality with smoothing~\cite{polyanskiy2025information}. Define the smoothed distribution
\[
Q_k := (1-\alpha_k) P_k + \alpha_k U,
\quad
U(c)=\frac{1}{|\mathcal C|},
\]
where
\[
\alpha_k := \sqrt{\frac{H(C\mid\mathcal H_{k-1})}{\log|\mathcal C|}}.
\]
According to the definition of  the total variation, we have
\[
P_k(c^\star) - Q_k(c^\star) \le \|P_k-Q_k\|_{\mathrm{TV}}.
\]

For the event $\{c^\star\}$, total variation yields
\[
1-P_k(c^\star)
\le
1-Q_k(c^\star) + \|P_k-Q_k\|_{\mathrm{TV}}.
\]
Since $Q_k(c^\star) =  (1-\alpha_k) P_k(c^\star) + \frac{\alpha_k}{|\mathcal C|}$ and $P_k(c^\star)>0$ for $c^\star$ is the true context, then $Q_k(c^\star)  \ge \alpha_k/|\mathcal C|$ and
\begin{equation*}
    \begin{aligned}
        \|P_k-Q_k\|_{\mathrm{TV}}& =\|P_k- (1-\alpha_k) P_k - \alpha_k U\|_{\mathrm{TV}} \\
        & = \alpha_k\|P_k-U\|_{\mathrm{TV}}, 
    \end{aligned}
\end{equation*}
we obtain
\[
1-P_k(c^\star)
\le
1-\frac{\alpha_k}{|\mathcal C|}
+
\alpha_k\|P_k-U\|_{\mathrm{TV}}.
\]
Applying Pinsker's inequality,
\[
\alpha_k\|P_k-U\|_{\mathrm{TV}} \le \sqrt{D_{\mathrm{KL}}(P_k\|U)}.
\]
Note that
$D_{\mathrm{KL}}(P_k\|U)=\log|\mathcal C|-H(C|\mathcal H_{k-1})$, we obtain
\[
1-P_k(c^\star)
\le
1-\frac{\alpha_k}{|\mathcal C|}
+
\alpha_k\sqrt{\tfrac12\big(\log|\mathcal C|-H(C\mid\mathcal H_{k-1})\big)}.
\]
Using $\sqrt{a-b}\le\sqrt a$ when $b\ge 0$ and the definition of $\alpha_k$,
\begin{equation*}
\begin{aligned}
    \alpha_k\sqrt{\tfrac12(\log|\mathcal C|-H(C\mid\mathcal H_{k-1}))}
&\le \alpha_k  \sqrt{\frac{1}{2} \log|\mathcal C|} \\
&=\sqrt{\tfrac12\,H(C\mid\mathcal H_{k-1})},
\end{aligned}
\end{equation*}
and since $\alpha_k = \sqrt{H(C|\mathcal H_{k-1})/\log|\mathcal C|}$, the remaining term is of the
same order. Therefore, there exists a universal constant $\xi>0$ such that
\begin{equation}
\label{eq:1mp_bound}
1-P_k(c^\star)
\le
\xi \sqrt{H(C\mid\mathcal H_{k-1})}.
\end{equation}
Also since the $\mathcal{C}$ is finite, the value of $H(C|\mathcal H_{k-1})$ must be finite. And the episode number  is finite, there must exist a universal constant $\eta>0$ such that
\begin{equation*}
H(C\mid\mathcal H_{k-1}) \le \frac{1}{\eta} I(C;Y^k \mid \mathcal{H}_{k-1}). 
\end{equation*}
Then similarily, by the chain rule of mutual information,
\begin{equation}
\label{eq:entropy_sum_bound}
\sum_{k=1}^K H(C\mid\mathcal H_{k-1}) \le \frac{1}{\eta} \sum_{k=1}^K I(C;Y^k \mid \mathcal H_{k-1})
= \frac{\log|\mathcal C|}{\eta}.
\end{equation}
Therefore,
\begin{equation}
\begin{aligned}
\sum_{k=1}^K  I_k^{(1)} &\le 2T R_{\max} \sum_{k=1}^K  \big(1 - P_k(c^\star)\big)\\
& \le \xi T R_{\max} \sum_{k=1}^K  \sqrt{H(C\mid\mathcal H_{k-1})}\\
& \le \xi T R_{\max} \sqrt{ K \sum_{k=1}^K H(C\mid\mathcal H_{k-1})} \\
& \le \xi T R_{\max} \sqrt{\frac{K \log|\mathcal C|}{\eta}}.
\end{aligned}
\end{equation}
Then first inequality is by~\eqref{eq:i1_1mp_bound}, the second inequality is by~\eqref{eq:1mp_bound}, the third inequality is by Cauchy–Schwarz inequality, and the last inequality is by~\eqref{eq:entropy_sum_bound}.

\paragraph{Proof of Lemma~\ref{lem:I2}} We first prove that $\{I_k^{(2)}\}$ forms a martingale difference sequence. 

\begin{lemma}
For all \( k \),
\[
\Expect[I_{k}^2 \mid \mathcal H_{k-1}] = 0.
\]
\end{lemma}

\begin{proof}
Since \( \pi_k \) is \( \mathcal H_{k-1} \)-measurable,
\[
\bar V_k(\pi)
:= \Expect_{c\sim P(C| \mathcal{H}_{k-1})} \!\left[V_c(\pi) \right],
 \]
Therefore,
\[
\Expect[I_k^{(2)} \mid \mathcal H_{k-1}]
= \Expect[\bar V_k(\pi_k) \mid \mathcal{H}_{k-1}] - \mathbb{E}[V_{c^\star}(\pi_k) \mid \mathcal{H}_{k-1}].
\]
By definition, $\bar V_k(\pi) := \mathbb{E}_{c\sim P(\cdot| \mathcal{H}_{k-1})} [V_c(\pi)]$. Thus:
\[
\mathbb{E}[\Delta_k \mid \mathcal{H}_{k-1}]
= \bar V_k(\pi_k) - \bar V_k(\pi_k) = 0.
\]
\end{proof}

Therefore, $\{I_k^{(2)}\}$ forms a martingale difference sequence.

Let the realized episodic return in episode $k$ be
\[
G_k := \sum_{t=0}^{T-1} r_{c^\star}(S_t^{k},A_t^{k}),
\quad
\bar G_k := \Expect[G_k \mid \mathcal H_{k-1}],
\]
and define
\[
I_k^{(2)} := \bar G_k - G_k.
\]

Since $\bar G_k=\mathbb E[G_k\mid \mathcal H_{k-1}]$,
\[
\Expect[I_k^{(2)} \mid \mathcal H_{k-1}]
=
\Expect[\bar G_k-G_k\mid \mathcal H_{k-1}]
=
\bar G_k-\bar G_k
=
0,
\]
so $\{I_k^{(2)}\}_{k=1}^K$ is a martingale difference sequence w.r.t.\ $\{\mathcal H_k\}$.

If $|r_c(s,a)|\le R_{\max}$, then $|G_k|\le TR_{\max}$ and $|\bar G_k|\le TR_{\max}$, hence
\[
|I_k^{(2)}|
=
|\bar G_k-G_k|
\le
|\bar G_k|+|G_k|
\le
2TR_{\max}.
\]

By the Azuma--Hoeffding inequality~\cite{rakhlin17a},
\begin{equation}
\Expect \left|\sum_{k=1}^K I_k^{(2)}\right|\le 2TR_{\max}\sqrt{2K}
\end{equation}

\section{Inference Mechanism}
\label{app:inference}
\subsection{Observable Operator}


Let $N =\lvert{\calS \rvert}$ be the total number of product states in $M$. We introduce an indexing of the product state set $\calS$.  
The transition function $P$ in the CPOMDP with domain $\calS \times \calA$ and codomain $\dist{\mathcal{S}}$ is directly extended to the domain $\{1,\ldots, N\}\times \mathcal{A}$ and co-domain $\dist{\{1,\ldots, N\}}$  using the state indices.

Consider an action $a\in \calA$, the $i,j$-th entry of the transposed transition matrix $\matT^a \in \reals^{N\times N }$   is defined to be:
\[
\matT_{i,j}^a = P(i|j,a), \forall i , j \in \{1,\ldots, N\} 
\]
which is the probability of reaching state $i$ from state $j$ with action $a$.
Similarly, we index the observation set $\cal O$ as $\{1,\ldots, M\}$ and refer to each $o \in \calO$ by its index. 
For each $a\in \calA$, Let $\matO^a \in \reals^{M \times N}$ be the observation probability matrix with $\matO_{o,j}  =E( o |  j)$. 
\begin{definition} 
Given the CPOMDP $M$, for any pair of observation and  action $(o,a)$,
 the observable operator given  action $\matA_{o|a}$ is a matrix of size $N \times N$ with its $ij$-th entry defined as $$
 \matA_{o |a }[i,j] =  \matT_{i , j}^a \matO_{o,j}  \ ,$$
 which is the probability of transitioning from state $j$ to state $i$ after taking action $a$ and at the state $j$,  an observation $o$ is emitted.
  In matrix form, 
\[
\matA_{o |a } = \matT^a \text{diag}(\matO_{o, 1} , \dots, \matO_{o , N} ).
\]
\end{definition}
From~\cite{wei2025activeinferenceincentivedesign}, the probability
\begin{equation}
\label{eq:matrix_operation_c}
P(o_{0:t} | a_{0:t}, i) = \mathbf{1}_N^\top \matA_{o_t|a_t} \dots \matA_{o_0|a_0} \mu_i.
\end{equation}
given the context $i$. And then 
\begin{equation}
P_\theta(y |i)  =  \frac{P(o_{0:t} | a_{0:t}, i) }{P(o_{0} | a_{0}, i) }  \prod_{i=0}^{t} \pi_\theta(a_{i}| o_{0:i-1}).
\end{equation}

\subsection{Extended Kalman Filter}
\label{subapp:EKF}

For the convenience and efficiency, to implement the algorithm using extended Kalman filter, we introduce a observation indicator $m_t$ where $m_t = 1$ if the action $a_t \neq o$. That means the observation is masked. Let
\[
y_{1:T} := (a_{1:T}, z_{1:T}, m_{1:T})
\]
denote a trajectory consisting of actions $a_t$, observations $z_t$, and observation indicators $m_t \in \{0,1\}$.
Conditioned on context $c$, we approximate the belief over the continuous latent state $x_t$ by a Gaussian
\[
x_t \mid y_{1:t-1}, c \sim \mathcal N(\mu_t, \Sigma_t),
\]
initialized as $(\mu_1, \Sigma_1) = (\mu_{0,c}, \Sigma_{0,c})$.

Given action $a_t$, the predicted belief is
\[
\mu_{t|t-1} = \mu_t + f(a_t), \qquad
\Sigma_{t|t-1} =
\begin{cases}
\Sigma_t + Q, & a_t \in \{\texttt{l}, \texttt{r}\}, \\
\Sigma_t, & a_t = \texttt{o},
\end{cases}
\]
where $f(\texttt{l})=-\Delta$, $f(\texttt{r})=\Delta$, $f(\texttt{o})=0$, and $Q$ denotes the process noise variance.

If an observation is acquired ($m_t=1$ and $a_t=\texttt{o}$), we assume the observation model
\[
z_t = x_t + v_t, \qquad v_t \sim \mathcal N(0, \mu_{t|t-1}),
\]
where $R_c(\cdot)$ is the context-dependent observation variance.
The Kalman update is given by
\[
K_t = \frac{\Sigma_{t|t-1}}{\Sigma_{t|t-1} + \mu_{t|t-1}},
\]
\[
\mu_{t+1} = \mu_{t|t-1} + K_t (z_t - \mu_{t|t-1}), \qquad
\Sigma_{t+1} = (1 - K_t)\Sigma_{t|t-1}.
\]
If no observation is available, the belief remains at the predicted state.

The trajectory likelihood conditioned on context $c$ is
\[
P(y_{1:T} \mid c)
=
\prod_{t=1}^T
\pi_\theta(a_t \mid h_t)
\;
\Big[
\mathcal N\!\bigl(z_t;\mu_{t|t-1}, \mu_{t|t-1}\bigr)
\Big]^{\mathbb I\{m_t=1,\; a_t=\texttt{o}\}},
\]
where $h_t = (z_{1:t-1}, m_{1:t-1})$ denotes the observation history.

Equivalently, the log-likelihood is
\[
\log P(y_{1:T} \mid c)
=
\sum_{t=1}^T \log \pi_\theta(a_t \mid h_t)
+
\sum_{t=1}^T
\mathbb I\{m_t=1,\; a_t=\texttt{o}\}
\log \mathcal N\!\bigl(z_t;\mu_{t|t-1}, \mu_{t|t-1}\bigr).
\]

\section{Summary of Variational Policy Gradient Algorithm}
\label{app:summary}

\begin{algorithm}[H]
\caption{Variational CPOMDP Solver}
\label{alg:cpomdp_solver}
\begin{algorithmic}[1]
\REQUIRE Context set $\mathcal C$, policy class $\Pi_\theta$, horizon $T$, temperature $\tau$, number of episodes $K$

\STATE Initialize policy parameters $\theta^{0}$

\FOR{$k = 1$ to $K$}
    \STATE Sample context $c^{k} \sim P(C)$
    \STATE Execute policy $\pi_{\theta^{k-1}}$ in context $c^{k}$ to obtain observations
    \[
        y^{k} = (o_0^{k}, a_0^{k}, \dots, o_T^{k})
    \]
    \STATE Compute posterior belief $P(C \mid y^{k})$
    \STATE Compute entropy-regularized return
    \[
        \mathcal J(y^{k}, c^{k})
        =
        \log P_\theta(y^{k})
        -
        \frac{R_{c^{k}}(y^{k})}{\tau}
    \]
    \STATE Estimate policy gradient
    \[
        \nabla_\theta \mathcal L(\theta^k)
        \approx
        \mathcal J(y^{k}, c^{k})
        \sum_{t=0}^T
        \nabla_\theta
        \log \pi_{\theta^k}(a_t^{k} \mid o_{0:t-1}^{k})
    \]
    \STATE Update policy parameters
    \[
        \theta^{k+1}
        \gets
        \theta^{k}
        +
        \alpha \nabla_\theta \mathcal L(\theta^{k})
    \]
\ENDFOR

\STATE \textbf{return} $\pi_{\theta^\star} \coloneqq \pi_{\theta^{K}}$
\end{algorithmic}
\end{algorithm}

\begin{figure*}[t]
  \vskip 0.2in
  \centering

  \subfigure[C-IDS policy under context 0\label{fig:policy_ctx0}]{
    \includegraphics[width=0.45\textwidth]{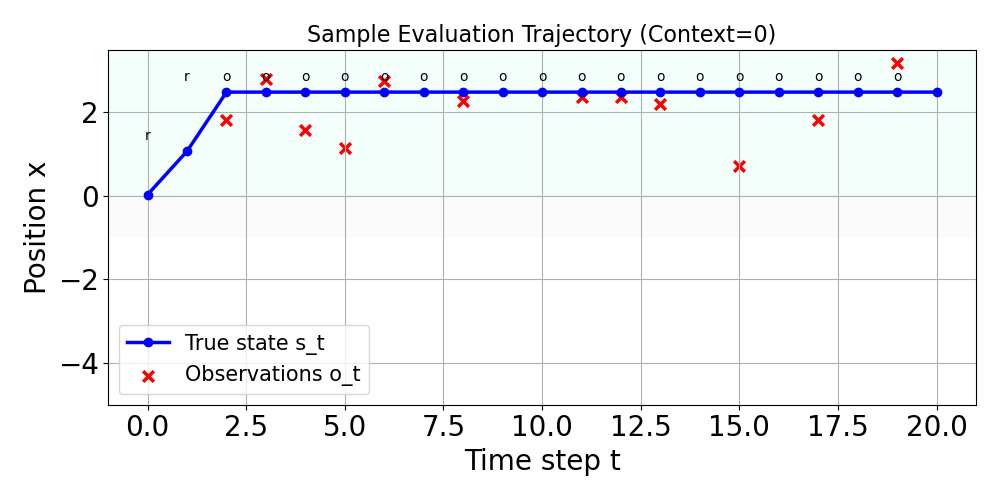}
  }
  \hfill
  \subfigure[C-IDS policy under context 1\label{fig:policy_ctx1}]{
    \includegraphics[width=0.45\textwidth]{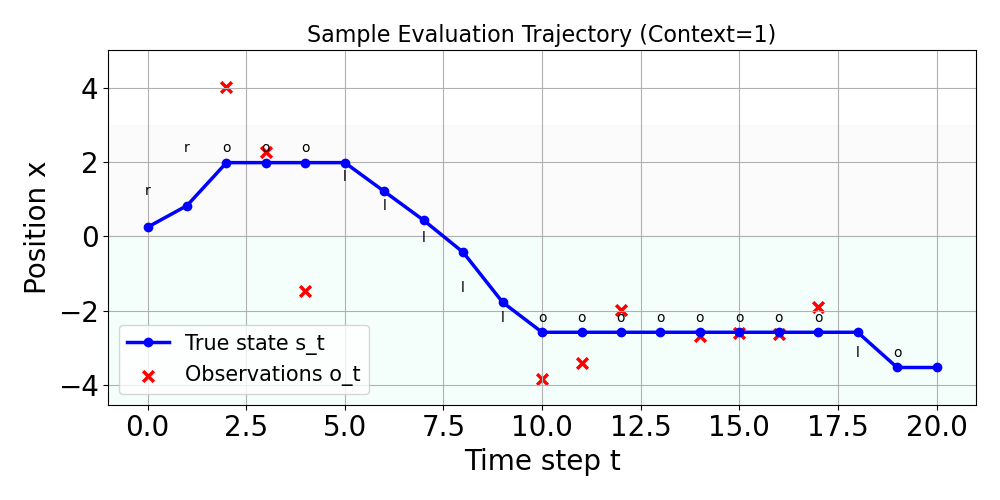}
  }

  \vskip 0.15in

  \subfigure[Policy comparison (context 0)\label{fig:policy_compare_ctx0}]{
    \includegraphics[width=0.45\textwidth]{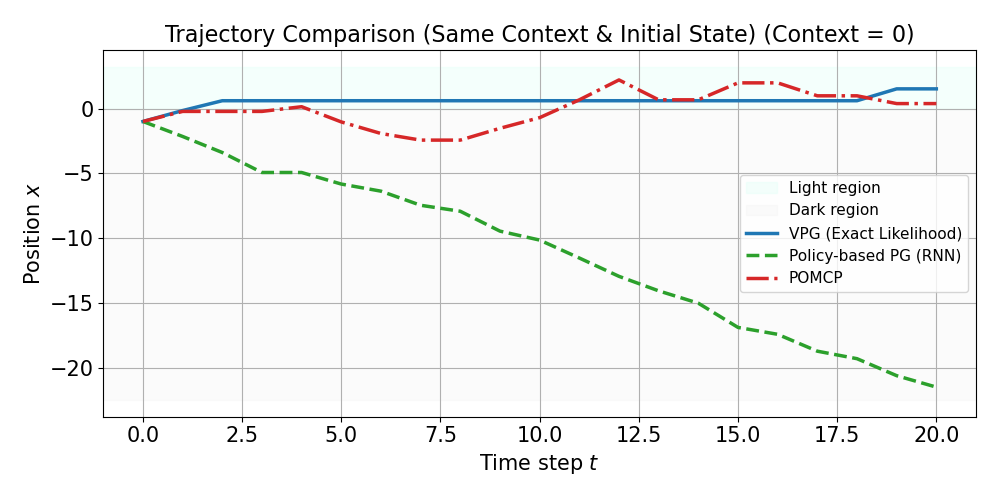}
  }
  \hfill
  \subfigure[Policy comparison (context 1)\label{fig:policy_compare_ctx1}]{
    \includegraphics[width=0.45\textwidth]{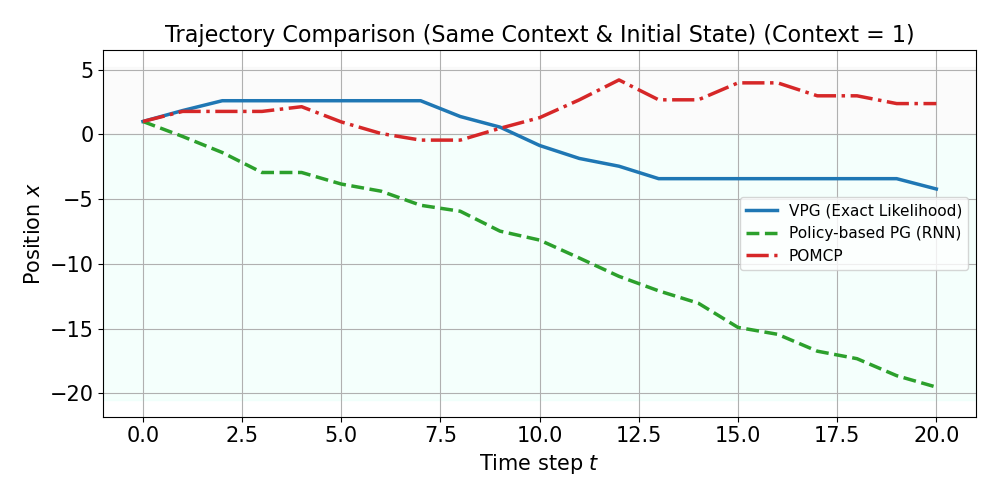}
  }

  \caption{Trajectories generated under C-IDS policy and comparisons with baseline methods.}
  \label{fig:policy}
\end{figure*}

\section{Visualization Results}
\label{app:visualization}
We generate some trajectories under the learned optimal policy, as shown in Figures~\ref{fig:policy_ctx0} and~\ref{fig:policy_ctx1}. The agent initially moves to the right in both cases to gather informative observations. In context~$0$ (Figure~\ref{fig:policy_ctx0}), this behavior directly leads the agent into the reward region, allowing it to obtain positive returns. In contrast, in context~$1$ (Figure~\ref{fig:policy_ctx1}), the agent detects the increased observation variance through active sensing, correctly infers the underlying context, and subsequently reverses direction to move left and reach the reward region.

We further visualize trajectories generated under different policies in Figures~\ref{fig:policy_compare_ctx0} and~\ref{fig:policy_compare_ctx1}. As shown, the POMCP policy consistently moves to the right, while the RDPG--RNN policy consistently moves to the left, regardless of the underlying context. This behavior indicates that both baseline methods fail to detect and adapt to the latent context, resulting in policies that commit to a single direction. In contrast, C-IDS dynamically adjusts its actions by inferring the context from observations, enabling context-aware directional decisions and improved task performance.


\end{document}